\documentclass[a4paper,10pt]{article}

\usepackage[english]{babel}
\usepackage{authblk}

\usepackage{changes, pdfcomment}

\usepackage{amsmath, amsfonts, amsthm, amssymb, mathtools, nicefrac }
\newtheorem{lemma}{Lemma}
\usepackage[subnum]{cases}

\usepackage{graphicx, subfig}
\usepackage{pgfplots, pgfplotstable}
\usepackage{tikz}
\usepackage{xcolor, color, colortbl, pifont}

\usepackage{float, enumerate, array, caption, framed, enumitem, listings, tabularx,
			makecell, url}
\usepackage{eurosym, bbm, gensymb }
\usepackage{breakurl}

\usepackage{appendix, booktabs, lscape, selectp}
\usepackage[vlined, boxruled]{algorithm2e}

\usepackage[backend=bibtex, hyperref=false, natbib=true, style=chicago-authordate, backref=true, backrefstyle=three,  firstinits=true]{biblatex}

\usepackage{hyperref}
\hypersetup{
	pdfauthor = {T.R.B. den Haan, K.W. Chau, M. van der Schans, C.W. Oosterlee}, 
	pdftitle = {Rule-based Strategies for Dynamic Life Cycle Investment},
	pdfkeywords = {Life cycle investing} {Pensions} {Defined contribution} {Rule-based strategies} {Dynamic programming},
    breaklinks=true
}

\title{Rule-based Strategies for Dynamic Life Cycle Investment}
\date{\today} 

\author[1,3]{T.R.B.~den~Haan%
	\thanks{Electronic address: \texttt{Tim.denHaan@ortec-finance.com}; Corresponding author}
}
\author[4]{K.W.~Chau}
\author[3]{M.~van~der~Schans}
\author[1,2]{C.W.~Oosterlee}
\affil[1]{Delft University of Technology, DIAM - Delft Institute of Applied Mathematics, Delft, the Netherlands}
\affil[2]{CWI - The Center for Mathematics and Computer Science, Amsterdam}
\affil[3]{Ortec Finance, Rotterdam, the Netherlands}
\affil[4]{University of Groningen, Department of Economics, Econometrics and Finance, Groningen, the Netherlands}

\DeclareMathOperator*{\argmax}{arg\,max}

\newcommand{\R}{\mathbb{R}}
\newcommand{\E}{\mathbb{E}}
\newcommand{\T}{\mathcal{T}}
\newcommand{\A}{\mathcal{A}}
\newcommand{\F}{\mathcal{F}}
\renewcommand{\S}{\mathcal{S}}

\newcommand{\EX}[1]{\E\left[ #1 \right]}

\newcommand{\lrp}[1]{\left( #1 \right)}
\newcommand{\lrb}[1]{\left[ #1 \right]}

\pgfplotsset{
	/pgfplots/area legend/.style={%
		/pgfplots/legend image code/.code={%
			\fill[##1] (0cm,-0.1cm) rectangle (0.6cm,0.1cm);
		}%
	},
}

\definecolor{OFblue}{RGB}{0,132,203}
\definecolor{OForange}{RGB}{245,128,37}
\definecolor{OFgreen}{RGB}{135,187,64}
\definecolor{OFdarkblue}{RGB}{33,64,122}
\definecolor{OFyellow}{RGB}{252,175,67}
\definecolor{OFmediumblue}{RGB}{106,117,161}
\definecolor{OFdarkblue}{RGB}{33,64,122}
\definecolor{OFdarkgray}{RGB}{168,169,172}
\definecolor{OFlightgray}{RGB}{208,211,214}
\definecolor{CustomPurple}{RGB}{229,100,100}
\definecolor{PyBlue}{rgb}{0.12156862745098,0.466666666666667,0.705882352941177}
\definecolor{PyOrange}{rgb}{1,0.498039215686275,0.0549019607843137}
\definecolor{PyGreen}{rgb}{0.172549019607843,0.627450980392157,0.172549019607843}
\definecolor{PyRed}{rgb}{0.83921568627451,0.152941176470588,0.156862745098039}
\definecolor{PyPurple}{rgb}{0.580392156862745,0.403921568627451,0.741176470588235}
\definecolor{PyBrown}{rgb}{0.549019607843137,0.337254901960784,0.294117647058824}
\definecolor{PyPink}{rgb}{0.890196078431372,0.466666666666667,0.76078431372549}
\definecolor{PyYellowGreen}{RGB}{245,148,0}
\definecolor{TUblue}{RGB}{0,166,214}
\definecolor{TUblack}{RGB}{0,0,0}
\definecolor{TUwhite}{RGB}{255,255,255}
\definecolor{TUskyblue}{RGB}{110,187,213}
\definecolor{TUpurple}{RGB}{29,28,115}
\definecolor{TUorange}{RGB}{230,70,22}
\definecolor{TUyellow}{RGB}{225,196,0}
\definecolor{TUred}{RGB}{226,26,26}
\definecolor{TUgreen}{RGB}{0,136,145}
\definecolor{TUbrightgreen}{RGB}{165,202,26}
\definecolor{TUwarmpurple}{RGB}{109,23,127}
\definecolor{TUgreygreen}{RGB}{107,134,137}
\definecolor{color0}{rgb}{0.0439830834294502,0.489042675893887,0.257977700884275}
\definecolor{color1}{rgb}{0.201307189542484,0.644444444444444,0.338562091503268}
\definecolor{color2}{rgb}{0.468896578239139,0.771318723567858,0.395770857362553}
\definecolor{color3}{rgb}{0.686274509803921,0.866205305651672,0.438523644752018}
\definecolor{color4}{rgb}{0.862668204536717,0.942176086120723,0.561091887735486}
\definecolor{color5}{rgb}{0.999923106497501,0.99761630142253,0.745021145713187}
\definecolor{color6}{rgb}{0.996386005382545,0.887966166858901,0.561091887735486}
\definecolor{color7}{rgb}{0.992848904267589,0.716955017301038,0.409457900807382}
\definecolor{color8}{rgb}{0.966551326412918,0.497424067666282,0.295040369088812}
\definecolor{color9}{rgb}{0.881045751633987,0.26797385620915,0.189542483660131}
\definecolor{color10}{rgb}{0.731641676278355,0.0811995386389852,0.150711264898116}

\pgfplotscreateplotcyclelist{MyLineList}{%
	{OFblue, mark = none, thick},
	{OForange, mark = none, thick},
	{OFgreen, mark = none, thick},
	{OFyellow, mark = none, thick},
	{OFdarkblue, mark = none, thick},
	{TUblue, mark = none, thick},
	{TUblack, mark = none, thick}
}

\pgfplotscreateplotcyclelist{MyDashedList}{%
	{OFblue, mark = none, thick},
	{OForange, mark = none, thick, dashed},
	{OFgreen, mark = none, thick, dashed},
	{OFyellow, mark = none, thick, loosely dashed},
	{OFdarkblue, mark = none, thick, dashed},   
	{TUblue, mark = none, thick, dashed},
	{TUblack, mark = none, thick, dashed}
}
\pgfplotscreateplotcyclelist{MyScatterList}{%
	{only marks, draw=OFblue, fill=OFblue},
	{only marks, draw=OForange, fill=OForange},
	{only marks, draw=OFgreen, fill=OFgreen},
	{only marks, draw=OFdarkblue, fill=OFdarkblue},
	{only marks, draw=OFyellow, fill=Ofyellow},
	{only marks, draw=TUblue, fill=TUblue},
	{only marks, draw=Tublack, fill=TUblack}
}
\pgfplotscreateplotcyclelist{MyFillList}{%
	{fill=OFblue},
	{fill=OForange},
	{fill=OFgreen},
	{fill=OFyellow},
	{fill=OFmediumblue},
	{fill=OFdarkgray},
	{fill=OFlightgray},
	{fill=OFdarkblue}
}
\pgfplotsset{
	colormap={MyColorMap}{
		rgb=(0.0439830834294502,0.489042675893887,0.257977700884275),
		rgb=(0.201307189542484,0.644444444444444,0.338562091503268),
		rgb=(0.468896578239139,0.771318723567858,0.395770857362553),
		rgb=(0.686274509803921,0.866205305651672,0.438523644752018),
		rgb=(0.862668204536717,0.942176086120723,0.561091887735486),		
		rgb=(0.999923106497501,0.99761630142253,0.745021145713187),
		rgb=(0.996386005382545,0.887966166858901,0.561091887735486),
		rgb=(0.992848904267589,0.716955017301038,0.409457900807382),
		rgb=(0.966551326412918,0.497424067666282,0.295040369088812),
		rgb=(0.881045751633987,0.26797385620915,0.189542483660131),
		rgb=(0.731641676278355,0.0811995386389852,0.150711264898116),
	},
}

\begin{document}

\maketitle

\section{Introduction}\label{introduction}
Nowadays, many people invest their retirement savings in a defined contribution pension scheme. In such a scheme, the contributions are agreed upon and are, e.g., a percentage of one's salary. The pension, however, is uncertain as it depends on the returns on investment. At retirement, the accumulated wealth is converted to a pension income that intends to replace a proportion of the investor's income, typically about 70\%, which is referred to as the replacement ratio. In this paper, we propose a dynamic strategy that optimally steers the investor towards a replacement ratio target.

Our dynamic strategy will reduce risk after several years of good returns on investment. It presumes that upward potential concurs with downside risk. Our pension investor is only interested in reaching her replacement ratio target, i.e., not making the target is considered downside risk and she feels indifferent about any two values above the target. We will show that, in this sense, the designed dynamic strategy outperforms static life cycle strategies. By decreasing risk after several good years, our dynamic strategy prevents unnecessary risk taking.

A well-known static life cycle strategy is known as Bogle's rule \citep{bogle}, which prescribes to invest $100\%$ minus one's age in risky assets. Decreasing risk in the course of the life cycle in such a way is called a glide path. When the glide path is known in advance up to retirement, the strategy is static and does not adjust as events unfold. Therefore, static strategies may take unnecessary risk when returns on investment are better than anticipated, see \citet{Arnott2013, Graf2017} for a discussion of drawbacks of static life cycle strategies. The strategy we propose is also rule-based, but it is dynamic as the prescribed rule depends on events that still have to unfold.

In the literature, dynamic strategies are often studied in the context of dynamic programming \citep{Bellman}. Dynamic programming optimizes the investment strategy backwards in time by optimizing decisions for the coming period given that consecutive decisions are already taken optimally.  \citet{Mer69} was the first to apply dynamic programming to an asset allocation problem with two assets, a risky and a risk-free asset, also allowing for consumption during the investment period. Optimal decisions were based on the constant relative risk aversion utility function. \citet{Mer69} showed that the optimal strategy continuously rebalances, i.e., the optimal allocation is constant.

The literature on optimal asset allocation is very rich, and we cite here some contributions that influenced our work.
The authors in \citet{Li2000} introduced mean-variance strategies with respect to a wealth target. The wealth target then allows the investor to identify a surplus: wealth up to the target may be invested in stocks, any remainder is invested in the risk-free rate. \citet{Zhang} solved a similar, however utility-based, problem and combined dynamic programming with the least squares Monte Carlo method. Upper and lower bounds for the wealth were prescribed in that paper, showing that upward potential comes with downside risk. Terminal wealth is steered towards a desired range by investing the difference between a risk-free-discounted upper bound value and the current wealth in the risk-free asset. 
\citet{Forsyth} also applied dynamic programming and used a PDE solver to solve a so-called time-consistent mean-variance problem, meaning that similar mean-variance problems were solved at future times. In addition to mean-variance that balances the mean and variance of returns, they studied a problem with a fixed wealth target. To reduce risk, both \citet{Forsyth} and \citet{Zhang} proposed to invest excess wealth in a risk-free asset. Similarly, the rule-based strategies introduced in this paper will invest excess wealth into a so-called matching portfolio. Compared to static strategies, distributions of outcomes are more centered around the target value and the area below the target value will become smaller. 

%TODO: insert some references to back these claims
Besides many positive aspects, dynamic programming and its resulting strategies also have some drawbacks. First of all, dynamic programming is computationally rather intensive. Secondly, the corresponding investment decisions can be sensitive to small changes in parameters and underlying assumptions. Because of this, the allocation may fluctuate over time resulting in large turnovers, of which, from a practical perspective, it is hard to explain why they are required. Intuitively defined rules, typically, do not suffer from these drawbacks. Moreover, it is not straightforward to apply dynamic programming to the pension settings as an investor's replacement ratio target often depends on inflation influencing the future, which in turn also influences the future contributions. 

Rule-based dynamic strategies fall in between the static and dynamic programming paradigms, when well constructed they aim for the best of both worlds. As shown by \citet{Basu2011}, even simple rule-based strategies that reduce risk half way in the life cycle can outperform static life cycle strategies. Compared to \citet{Basu2011}, our rule-based strategies can reduce risk annually, and consider the market price of future pension payments instead of a wealth target. Next to the rule-based strategies, we will also combine a rule-based strategy with dynamic programming in an integrated approach.

\section{The optimal asset allocation problem}
\subsection{Model setting}\label{theoretical_background}
To demonstrate the rule-based strategy's practical value, we will consider a specific pension investor (we will choose typical retirement data from the Netherlands). At $t=0$, the 26 year old investor will start saving up to retirement at time $t=T$, coinciding here with a retirement age of 67 years. She intends to replace 70\% of her income by her pension (including government allowances for old age). Although, in practice, an investor might be interested in insuring longevity risk or be interested in employing advanced withdrawal strategies, \citet{Blanchett2012} illustrates that simple withdrawal strategies can perform well, e.g., based on an annuity with a maturity roughly equal to an investor's life expectancy. Therefore, as we focus on accumulating wealth before retirement, we simply assume the investor buys an annuity that indexes with the expected inflation, i.e., a bond which, apart from indexation for expected inflation, equals annual payoffs, for a period of $N$, say 20, years after retirement. Whichever withdrawal strategy an investor might follow, the assumption here is that this annuity gives a good estimate of, at least, the investor's income in her first year after retirement, and, thereby, to what extend she can replace her salary for 70\% with a pension.

The investor can invest her wealth $W_t$ in a risky, equity-like, asset, which is called the return portfolio, or in a safe, bond-like, asset with annual payoffs during retirement, the matching portfolio. In our setting, the strategy will use the matching portfolio to protect the current gains, and it grows with inflation. Therefore, the matching portfolio also carries risk. Put differently, we assume the investor doesn't hedge inflation risk with inflation protected securities as the market for inflation protected is illiquid and strategies that hedge against inflation are not straightforward to follow in practice \citep{Martellini2014}. Finally, we assume there is no risk-free rate to invest money in.

The investor annually manages her portfolio, i.e., decisions, contributions and pension payments are made in discrete time, which runs up to retirement, from $t=0$ to $t=T$. The pension payments start at $t=T$ and run up to $t=T+N-1$. At time $t\leq T$ before retirement, she invests a fraction $\alpha_t$ of her wealth $W_t$ in the return portfolio. The investor is not allowed to short-sell assets or borrow money, so that
\begin{equation}
0\leq\alpha_t\leq1.
\end{equation}
In the dynamic programming literature, $\alpha_t$ is referred to as the control (as decisions intend to give the investor control over the outcome). A strategy maps information $Z_t$ available at time $t$, e.g., past returns and current wealth $W_t$, to the desired allocation:
\begin{equation}\label{eq:def_control}
\alpha_t:\R^K\ni Z_t \mapsto \alpha_t(Z_t) \in [0,1].
\end{equation}
Here $Z_t$ is adapted to a filtration $\F_t$, governing the underlying stochastic processes. Before time $t$, the information $Z_t$ is not yet available, and $\alpha_t$ is thus a stochastic quantity. In a static strategy, such as Bogle's rule, $\alpha_t$ only depends on time and is known, i.e., not stochastic, not even when the information $Z_t$ is not yet available. In practice, risk is reduced towards retirement, meaning that $\alpha_t$ typically decreases over time.

Just before rebalancing, the investor makes a contribution $c_t$ to the portfolio. These contributions resemble an age-dependent percentage $p_t$, see Table \ref{table_salary}, of the investor's salary $s_t$ which she earned in the period $t-1$ up to $t$. We assume that the investor's salary $s_t$ follows a deterministic career path, i.e., it increases with age. The investor's salary also increases stochastically with the wage inflation $w_t$, see Appendix \ref{sec:career-path-and-contribution-rate}.

%TODO: check if expected inflation should be added in discount rates
%TODO: how many contributions are there T or T+1
The investor's objective is to achieve her 70\% replacement ratio target at retirement without encurring too much downside risk. The replacement ratio at retirement, $R_T$, is given by
\begin{equation}\label{eqn:def_replacement_ratio}
R_T= \frac{W_T}{M_T} 
\cdot
\frac{T+1}{\sum\limits_{t = 0}^{T} s_t\prod\limits_{\tau=t+1}^{T} (1+\pi_\tau)}.
\end{equation}
Here, the second term divides the investor's average wage in nominal amounts indexed with inflation $\pi_t$ to retirement at $t=T$, and $M_t$ is the market value factor that discounts $N$ future pension payments indexed by expected inflation to time $t\leq T$:
\begin{equation}\label{eqn:market value factor}
M_t=\sum\limits_{\tau=T}^{T+N-1} (1+r^{\tau-t}_t)^{\tau-t}\;\E_t\left[\prod\limits_{\tau'=t+1}^\tau(1+\pi_{\tau'})\right],
\end{equation}
where $\E_t$ is the expectation, conditional on $\F_t$ (i.e., conditional on the information available at time $t$), and $r^{\tau-t}_t$ represents the market rates that discount payments from $\tau-t$ years into the future back to the present time. Using the market value factor $M_T$ at retirement, the first term in \eqref{eqn:def_replacement_ratio} converts the accumulated wealth $W_T$ to $N$ annual income payments indexed for expected future inflation. 

To measure whether a strategy achieves the investor's objective, we use a utility function, $U$, which, whenever decisions are to be taken, intends to maximize the following expression in expectation:
\begin{equation}
\max_{\alpha_t,\ldots, \alpha_{T-1}}~\E \left[{U(Z_T)}\left|\right.\F_t\right],
\end{equation}
where $\F_t$ represents current market information, $\alpha_t$ is as in \eqref{eq:def_control} and $Z_T$ is a vector with outcomes including the terminal replacement ratio.
Although other choices are possible, we choose $U(.)$ to be the shortfall below the investor's target replacement ratio of 70\%:
\begin{equation}\label{eqn:utility_shortfall}
U(Z_T)=\min(RR_T-70\%,0),
\end{equation}
where there is no shortfall in replacement ratio if it ends above $70\%$.
 Note that this measure is not conditional on the shortfall. 
So, additionally, we will also evaluate a strategy's performance using the 10\% conditional value at risk $\mathrm{CVaR_{0.1}}(RR_T)$ of the replacement ratio, i.e., the expectation of the $10\%$ worst case outcomes as defined by
\begin{equation}\label{kkk}
\mathrm{CVaR}_\alpha\; (RR_T)  = \E\left[RR_T\left|RR_T\leq F_{RR_T}^{-1}(\alpha)\right.\right],
\end{equation}
where $F_{RR_T}^{-1}(\alpha)$ is the inverse cumulative distribution function of terminal replacement ratio $RR_T$ and represents the $\alpha$-th quantile below which are the worst case outcomes.

\subsection{Governing stochastic model}\label{sec:governing-stochastic-model}
For general applicability, we require that the designed strategies are not defined in terms of the governing stochastic model parameters. That is, the strategies can be applied when different governing stochastic models would be used. We merely assume that the governing stochastic model can be simulated by means of a Monte Carlo simulation. To make this explicit, we choose to use a standard model developed to make risk analyses comparable between Dutch pension funds, see \citep{KNW2009}. The model and its calibration are well documented \citep{Draper2014}. Calibration on recent market data and a Monte Carlo simulation of the model are publicly available at the website of the Dutch Central Bank \citep{HBT2016}. In this paper, we use the set of 2017 (quarter 1), which is calibrated on data up to ultimo 2016 and start simulating from there.

In discrete time, the model is a VAR(1) model with normally distributed increments, see \citet{Muns2015} for a short summary of the model specification. In the calibration, some structure is imposed to achieve realistic market dynamics. Based on the model, sample paths are generated for the following variables:
\begin{itemize}
	\item Equity returns $x_t$, which are used for the return portfolio;
	\item Inflation $\pi_{t}$;
	\item Wage inflation $w_t$, which equals inflation $\pi_{t}$ plus $0.5\%$;
	\item A yield curve with interest rates $r_t^m$ containing rates for each maturity $m$.
\end{itemize}
The matching portfolio is tailored to the investor's retirement age. Its returns $m_t$ equal the rate of change in the market value factor:
\begin{equation}\label{eqn:returns matching portfolio}
m_t=\frac{M_t}{M_{t-1}}-1,
\end{equation}
where $M_t$ is defined in \eqref{eqn:market value factor}. Note that the matching portfolio protects the investor against \textit{expected} future inflation. To determine the expected future inflation, we use the least squares Monte Carlo technique, as presented in Section \ref{sec:least-squares-monte-carlo-method}.

Table \ref{tab:annual_stats} gives the annual return statistics of the variables. Due to the fluctuating market price of future pension payments, the standard deviation of the matching returns is very similar to the one of the equity returns. Although the matching portfolio follows these fluctuations, it is considered less risky, in terms of the investor's goals. By investing in the matching portfolio, the pensioner will receive the corresponding amount from the annuity, no matter the future market prices.
%TODO: should we quantify this?
\begin{table}[H]
	\centering
	\begin{tabular}{lrrrrr}
		{} &  $x_t$ &  $m_t$ & $r_t^{10}$ &$\pi_t$ &   $w_t$ \\
		\toprule
		Mean &	6.1\% & 3.4\% & 2.5\% & 1.6\% & 2.1\% \\
		Standard deviation&	18.3\% & 18.5\% & 2.4\% & 1.5\% & 1.5\% \\
		\midrule
		Correlations&&\\		
		Equity return ($x_t$)  &     1.00  &  &   &  &  \\
		Matching return ($m_t$)  &	-0.06 & 1.00  &  & & \\
		10 year interest ($r_t^{10}$) &	0.17  & -0.17 & 1.00  &  &  \\
		Inflation ($\pi_t$ )  &	0.11  & -0.04 & 0.82  & 1.00  & 1.00 \\
		Wage inflation ($w_t$) & 	0.11  & -0.04 & 0.82  & 1.00  & 1.00 \\
		\bottomrule
	\end{tabular}
	\caption{Annual statistics of the underlying stochastic model calculated on a sample that combines all sample paths.}\label{tab:annual_stats}%
\end{table}

\section{Rule-based strategies}\label{rule_based}
In this section, we define three rule-based strategies: a cumulative target strategy that decreases risk once it reaches a cumulative target for the contributions paid so far, an individual target strategy that tracks the investments of the contributions separately and decreases risk once it reaches the target for that contribution, and a combination strategy that combines the two with dynamic programming. The strategies all intend to steer towards a target replacement ratio of 70\%, and decrease risk when return on investment develops well. The strategies differ in their views on when return on investment has been developing well enough to decrease risk.

\subsection{Cumulative target strategy}\label{cumulative_target_strategy}
The cumulative target strategy that we consider here has similarities with the strategies studied in~\citet{Zhang} and~\citet{Forsyth}: risk is reduced once wealth exceeds a pre-defined wealth target. Contrary to~\citet{Zhang} and~\citet{Forsyth}, however, our investor saves for retirement and we relate the wealth target to the price of a bond with payoff equal to the desired pension.

Given a density forecast for the matching and return portfolios, see Section \ref{sec:governing-stochastic-model}, the strategy depends on two parameters: a required real rate of return $r$ (before retirement) and a discount rate $\delta$ (after retirement) to discount pension payments after retirement to the retirement date.
At time $t$ before retirement, i.e., $t\in 0\ldots T$, the investor contributes $c_t$ to her pension savings, see Table \ref{table_salary}. The contributions $c_\tau$ up to time $t$, i.e., $\tau=0,\ldots,t$, are supposed to grow with inflation $\pi$, plus the real rate of return $r$, to a target wealth $c_\tau\; \E_t F_\tau$ at retirement, where $F_t$ is given by
\begin{equation}\label{eqn:target_wealth_factor}
F_\tau = \prod\limits_{\tau'=\tau+1}^T(1+r+\pi_{\tau'}),
\end{equation}
and the conditional expectation, $\E_t$, enforces that the realized inflation is used before time $t$ and  the expected inflation is used beyond time $t$. The wealth targets at retirement for all contributions $c_\tau$ up to time $t$ are combined and converted into a target pension using a discount factor, $\tilde M_T$, which is based on the discount rate $\delta$:
\begin{equation}\label{eqn:post_retirement_factor}
\tilde M_T=\sum\limits_{\tau=T}^{T+N-1} \frac{1}{(1+\delta)^{\tau-T}}.
\end{equation}
Using the market value factor $M_t$, as defined in \eqref{eqn:market value factor}, this gives us the following current target wealth $\tilde{W}_t$:
\begin{equation}\label{eqn:cumulative_target_wealth}
\tilde{W}_t=\frac{M_t}{\tilde M_T}\sum\limits_{\tau=0}^{t} c_\tau \; \E_tF_\tau,
\end{equation}
where the summation represents the combined wealth targets at retirement for all contributions $c_\tau$ up to time $t$.

The cumulative target strategy starts by investing new contributions $c_t$ in the risky asset. If the current wealth $W_t$, including the current contribution $c_t$, exceeds the target wealth $\tilde{W}_t$, risk is reduced and $W_t$ is transferred to the matching portfolio. For the matching portfolio, the investor follows a buy and hold strategy. New contributions invested in the risky asset, will also be transferred to the matching portfolio if the current wealth $W_t$, which consists of the current contribution $c_t$, the value of the matching portfolio and the value of the return portfolio, exceeds the target wealth $\tilde{W}_t$. In other words, at $t=0$, the control $\alpha_0$ is given by
\begin{numcases}{\alpha_0=\label{eqn:cumulative_control_alpha_0}}
0 &  if $W_0\geq\tilde W_0$,\\
1 & otherwise,
\end{numcases}
and, for $t=1\ldots T$, the control $\alpha_t$ is given by
\begin{numcases}{\alpha_t=\label{eqn:cumulative_control_alpha_t}}
0 &  if $W_t\geq\tilde W_t$,\\
\frac{\alpha_{t-1}(1+x_t)}{\alpha_{t-1}(1+x_t)+(1-\alpha_{t-1})(1+m_t)} & otherwise.
\end{numcases}

\subsection{Individual target strategy}\label{sec:indv_target}

Contrary to the cumulative target strategy, the individual target strategy, which is the second strategy we will analyze here, defines a wealth target per contribution and invests each contribution separately, i.e., the wealth $W_t$ is seen a sum of the individual wealth components resulting from investing the contributions separately:
\begin{equation}
W_t=\sum\limits_{\tau=0}^t W_{t,\tau}, 
\end{equation}
where $W_{t,\tau}$ is the wealth component from investing the contribution $c_\tau$.
As in \eqref{eqn:cumulative_target_wealth}, a wealth target $\tilde{W}_{t,\tau}$, at time $t$ for a contribution invested at time $\tau\leq t$, is given by
\begin{equation}\label{eqn:individual_wealth_target}
\tilde{W}_{t,\tau}=\frac{M_t}{\tilde M_T} c_\tau \; \E_tF_\tau.
\end{equation}
Apart from this, the strategy works similarly: the individual contributions are invested in the risky asset until the invested amount exceeds the wealth target for that contribution, in which case they are transferred to the matching portfolio until retirement. Thus, the control $\alpha_{t,\tau}$ for investing contribution $c_{t,\tau}$ is given by
\begin{numcases}{\alpha_{t,\tau}=\label{eqn:individual_control_alpha_t}}
0 &  if for any $\tau'=\tau\ldots t$ we have $W_{\tau',t} \geq \tilde W_{\tau',\tau},$\\
1 & otherwise.
\end{numcases}
At the aggregated level, the control $\alpha_t$ is now given by
\begin{equation}
\alpha_t=\frac{1}{W_t}\sum\limits_{\tau=0}^t W_{t,\tau}\alpha_{t,\tau}.
\end{equation}

Conceptually, the difference between the cumulative target strategy and the individual target strategy is what triggers the risk reduction. Contrary to the individual target strategy, in the cumulative target strategy new investments have to make up for insufficient past returns before a transfer to the matching portfolio can take place. On the other hand, in the cumulative target strategy good past returns may cause new contributions to be transferred immediately to the matching portfolio. With the individual target strategy, each contribution has to generate sufficient return on investment before such a transfer takes place.

\subsection{Combination strategy}\label{combination_strategy}
Both the cumulative and the individual target strategy either reduce risk by switching completely to the matching portfolio or don't reduce risk at all. Instead of completely switching or not switching at all, the combination strategy, which is the third strategy considered, combines the individual target strategy with dynamic programming to dynamically steer the wealth $W_{t,\tau}$ resulting from the contribution $c_\tau$ above its wealth target $\tilde W_{t,\tau}$. For this, we define the following wealth to target ratio,
\begin{equation}\label{eqn:state_variable}
Z_{t,\tau} := \frac{W_{t,\tau}}{\tilde W_{t,\tau}},
\end{equation}
and solve
\begin{equation}\label{eqn:dynamic_programming_problem_combination_strategy}
V(z,t,\tau) = \sup_{\A_{t,\tau}} \EX{\check{U}(Z_{T,\tau}) | Z_{t,\tau} = z},
\end{equation}
%\mathrm{such~that~} \A_i &= \{\alpha_t, \A^*_{i+1}\} := \{\alpha_i, \alpha^*_{i+1}, \ldots, \alpha^*_n\}\\
%\alpha_i &\in [0,1],
%\end{equation}}
where $\check{U}$ is a utility function, $V(z,t,\tau)$ is the value function in the dynamic programming problem and the control $\A_{t,\tau}$ consists of the future investment decisions:
\begin{equation}
\A_{t,\tau} = \{\alpha_{t,\tau}, \ldots,\alpha_{T,\tau}\}.
\end{equation}
Using the dynamic programming principle, it follows that the optimal control, $\A_{t,\tau}^*$, satisfies
\begin{equation}
\A_{t,\tau}^* = \{\alpha_{t,\tau}^*, \A_{t+1,\tau}^*\},
\end{equation}
which allows us to solve for the optimal control problem for $\A_{t,\tau}^*$, backwards in time.

In this context, we choose a utility function that steers the ratio $Z_{T,\tau}$ in between the bounds $z^*_{\mathrm{min}}$ and $z^*_{\mathrm{max}}$. This is in line with the investor's goal of minimizing downside risk, and with our assumption that upward potential comes with downside risk. The utility function should be positive concave and takes here the following functional form:
\begin{align}\label{eq:utility}
\check{U}(z) &= \frac{-\lrp{z-\beta}^2 - \lrp{z -  z^*_{\mathrm{min}}}^2}{z},
\end{align}
where
\begin{align*}
\beta &= \sqrt{2\lrp{z^*_{\mathrm{max}}}^2 - \lrp{z^*_{\mathrm{min}}}^2},
\end{align*}
see Figure \ref{utility_function} (note that this is a different utility function than $U(\cdot)$ from (\ref{eqn:utility_shortfall}). Utility function $\check{U}(\cdot)$ is clearly concave and continuous on the domain $\R_{>0}$. We set $z^*_{\mathrm{min}}=1$ and $z^*_{\mathrm{max}}=3$, as this choice fits well with the investor's replacement ratio target and, as we will show in Section \ref{sec:acarulebased}, is sufficient to demonstrate the strategy's added value.

\begin{figure}
	\centering
	\begin{tikzpicture}
  \begin{axis}[
    domain=0.1:6,
    samples=301,
    smooth,
    no markers,
    xlabel={State},
    ylabel={Utility},
    cycle list name = MyLineList,
    legend cell align={left},
    legend entries={Utility function, $z^*_{\mathrm{min}}$, $z^*_{\mathrm{max}}$},
    tick align=outside,
    legend pos = south east,
    xtick = {0,1,2,3,4,5,6},
    tick pos=left,
    width=\linewidth,
    height = 0.3\textheight,
    xmin = 0, xmax = 6,
    ymin = -25, ymax = 0,
    ]
    \addlegendimage{no markers, OFblue}
    \addlegendimage{dotted, OForange}
    \addlegendimage{dotted, red}
    \addplot {(-(x-sqrt(2*3^2-1^2))^2 - (x-1)^2)/x};
    \addplot [semithick, OForange, dotted, forget plot]
    table [row sep=\\]{%
    	1	-1000 \\
    	1	1000 \\
    };
	\addplot [semithick, red, dotted, forget plot]
	table [row sep=\\]{%
		3	-1000 \\
		3	1000 \\
	};
	  \end{axis}
\end{tikzpicture}
	\caption{Plot of \eqref{eq:utility} with  $z^*_{\mathrm{min}}=1$ and $z^*_{\mathrm{max}}=3$.}
	\label{utility_function}
\end{figure}
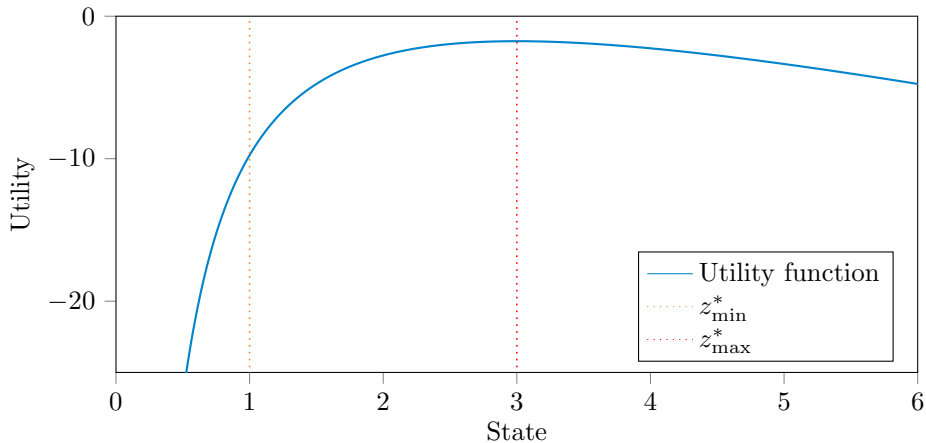

Now, we will show that the ratio $Z_{t,\tau}$, between the current wealth $W_{t,\tau}$ and its target $\tilde W_{t,\tau}$, evolves in time by making returns on investment in the nominator and updating the inflation expectation in the denominator.
Since this time evolution is independent of $\tau$, we can show that the optimal control $\alpha_{t,\tau}^*$ is independent of $\tau$, i.e., once the optimal control is found, it can be applied to all contributions.
\begin{lemma}\label{lemma:indepence_tau}
	The optimal control $\alpha_{t,\tau}^*$ of dynamic programming problem \eqref{eqn:dynamic_programming_problem_combination_strategy} is independent of the contribution $c_\tau$ and the time $\tau$ at which the contribution is made.
\end{lemma}
\begin{proof}
	The portfolio wealth $W_{t,\tau}$, accumulated by investing contribution $c_\tau$, increases with the return on investment and, therefore, satisfies
	\begin{equation}\label{eqn:recursing_wealth}
		W_{t,\tau}=\left[(1+x_t)\alpha_{t-1,\tau}+(1+m_t)(1-\alpha_{t-1,\tau})\right]W_{t-1,\tau}.
	\end{equation}
	From \eqref{eqn:individual_wealth_target}, \eqref{eqn:target_wealth_factor} and \eqref{eqn:returns matching portfolio}, it follows that the wealth target, $\tilde W_{t,\tau}$, satisfies
	\begin{equation}\label{eqn:recursing_wealth_target}
		\tilde W_{t,\tau}=\frac{\E_t F_{t-1}}{\E_{t-1}F_{t-1}} (1+m_t) \tilde W_{t-1,\tau}.
	\end{equation}
	Substitution of \eqref{eqn:recursing_wealth} and \eqref{eqn:recursing_wealth_target} in \eqref{eqn:dynamic_programming_problem_combination_strategy} yields that the optimal controls $\alpha_{t,\tau}^*$ solve
	\begin{equation}\label{eqn:optimal_control}
	\sup_{\A_{t,\tau}} \EX{\left.
		\check{U}\left(
		z \prod\limits_{\tau'=t+1}^T
		\frac{(1+x_{\tau'+1})\alpha_{\tau',\tau}+(1+m_{\tau'+1})(1-\alpha_{\tau',\tau})}{\frac{E_{\tau'+1} F_{\tau'}}{E_{\tau'} F_{\tau'}} (1+m_{\tau'+1}) }
		\right)
		\right| Z_{t,\tau} = z}.
	\end{equation}
	This shows that both the value function $V(z,t,\tau)$ and the optimal control $\alpha_{t,\tau}^*$ are independent of $\tau$.
\end{proof}
Lemma \ref{lemma:indepence_tau} implies that, theoretically, the dynamic programming problem has to be solved only once, i.e., the investment decisions for the first contribution $c_0$ can be used for all other contributions.

For the practical implementation for the dynamic programming algorithm, readers may refer to Appendix \ref{sec:algorithm}.

\subsection{Target replacement ratio} \label{sec:target-replacement-ratio}

The variable $r$, used in the construction of the wealth target, can be interpreted in multiple ways. First of all, it serves as a discount rate, which is used to compute the present value of contributions that are made in the future. It can also be viewed as an annual return requirement: each contribution is required to have an average annual return of $r$. A third interpretation of $r$ is that of a future expected annual return. The computation of the expected replacement ratio requires a future annual return assumption.

Let $t\in\T$ and let $\F_t$ be the corresponding filtration. The expected replacement ratio $R_t$ is defined as

\begin{align*}
R_t &:= \frac{\EX{P | \F_t}}{\EX{\sum\limits_{t = 0}^{T} s_t\prod\limits_{\tau=t+1}^{T} (1+\pi_\tau)|\F_t}},
\end{align*}	
where
\begin{align*}
\EX{P|\F_t} &= \frac{\EX{W_T|\F_t}}{\EX{M_T|\F_t}},
\end{align*}	
with
\begin{align*}
\EX{W_T|\F_t}&= \lrb{1+r+ I(T;t)}^{T-t}W_t + \sum_{k=t+1}^{T-1} \lrb{1+r+I(T; k})^{T-k}\EX{c_k|\F_t}\:.
\end{align*}

Computation of the expected replacement ratio requires four different estimators. The discount rate $r$ is used as an estimator for the future expected annual return. The estimator for the future inflation,  $I(T;t)$, has to be estimated through regression between the future and the past cumulative inflation, as shown in Equation (\ref{eqn:inflation}).
% The estimator for the future annual contributions is similar to the inflation estimator.
Future salaries are based on the information from Table \ref{table_salary}.
%Future salaries are based on estimating the future wage inflation. The age-dependent percentages of salaries that will be contributed in the future are assumed to be given. 
Lastly, the estimator for the market value factor at the end of the investment horizon, $\EX{M_T|\F_t}$, is based on regression between $M_t$ and $M_T$, with $\Phi = \{1,x\}$. 
See Appendix \ref{sec:least-squares-monte-carlo-method} for details of the regression method used.

The market value factor is considered to be independent of the discount rate $r$, inflation and wage inflation (the division operator can therefore be taken out of the expected value operator).

The computation of the target replacement ratio at time $t$ is similar to the computation of the expected replacement ratio. The only difference is that the portfolio wealth, $W_t$, is replaced by the target terminal wealth, $W^*(t)$. The target wealth definition causes the target replacement ratio, $R^*(t)$, to be independent of the market value factor:
\begin{align*}
R^*(t) &= \frac{\EX{P^*|\F_t}}{\EX{\sum\limits_{t = 0}^{T} s_t\prod\limits_{\tau=t+1}^{T} (1+\pi_\tau)|\F_t}},
\end{align*}
by using  the independence of the market factor $M_T$ and the wealth process and the definition of the current target wealth $\tilde{W}_t$,
\begin{align*}
\EX{P^*|\F_t} &= \frac{\EX{W^*(T)|\F_t}}{\EX{M_T|\F_t}} =\tilde{W}_t.
\end{align*}

As mentioned, the investor's target is to reach a replacement ratio of $70\%$. To translate this target into the wealth target in terms of portfolio wealth, we have:
\begin{align*}
\tilde{W}_t &= R^*(t)\EX{\sum\limits_{t = 0}^{T} s_t\prod\limits_{\tau=t+1}^{T} (1+\pi_\tau)|\F_t}\:.
\end{align*}

To steer towards a fixed replacement ratio target, $\tilde{W}(t)$ would have to be altered for each scenario. It is, however, easier to differ the quantity $R^*(t)$ slightly between scenarios, from a computational point of view. Instead, $\tilde{W}_t$ is defined as in Equation (\ref{eqn:cumulative_target_wealth}) and $r$ is set to the required annual return.

Numerically, we find that the target replacement ratio within a scenario is almost constant throughout time, as can be seen in the bottom-left plot of Figure \ref{single_scenario_cum}. Small alterations are caused by the estimators for the inflation and the wage inflation. Alterations of up to $0.01$ within a scenario are observed for a discount rate of $2.5\%$. Target replacement ratios are between $0.6847$ and $0.7033$ for a discount rate of $2.5\%$. 

\section{Numerical evaluation}
In this section, we apply the rule-based strategies described in Section \ref{rule_based} to the pension investor introduced in Section \ref{theoretical_background}, using the governing stochastic model described in Section \ref{sec:governing-stochastic-model}.

%TODO: continue here, review section below and revise discussion section

\subsection{Rule-based strategies}\label{sec:acarulebased}
To illustrate the dynamics of the rule-based strategies, Figure \ref{single_scenario_cum} shows one of the 2000 sample paths for the investor's portfolio dynamics. In particular, the top left figure shows the investor's wealth $W_t$ when following the cumulative target strategy (orange) and when the investor's wealth exceeds the target $\tilde W_t$ (yellow). Note that when this occurs, the investments are transferred to the matching portfolio (orange line, bottom right figure). The individual target strategy (in green) works similarly, but, as discussed, uses a target per contribution, so that, typically, only part of the wealth is transferred to the matching portfolio (green line, bottom right figure). The bottom left figure illustrates that, in this sample path, the rule-based strategies outperform the optimal static strategy in terms of expected replacement ratio - although only in the first 10 years of the investment the rule-based strategies take substantially more risk, i.e., have a substantially higher allocation to the return portfolio. Therefore, in this particular sample path, one could argue that the better performance comes from the rule-based strategies and not from increased exposure to risk.

\usepgfplotslibrary{groupplots}
\begin{figure}[tp!]
	\centering	
	\input{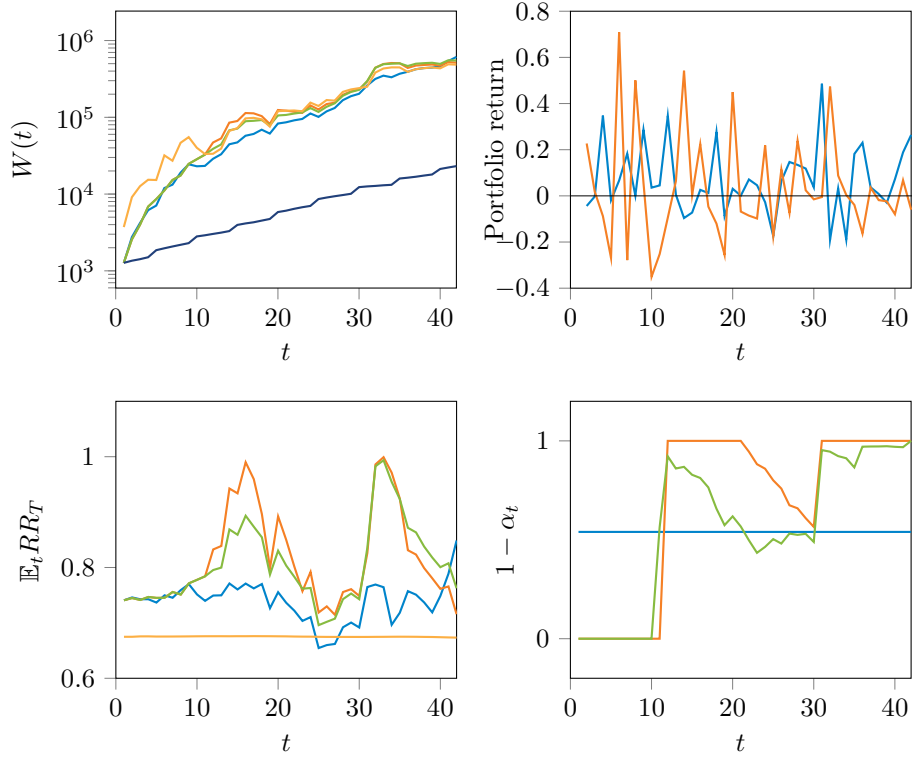}	
	\caption{Sample paths of wealth, return of the matching and return portfolio, expected replacement ratio, and allocation to the matching portfolio applied to the pension investor introduced in Section \ref{theoretical_background} following rule-based strategies with the discount rates $r=2\%$ and $\pi=2.5\%$, and using the governing stochastic model described in Section \ref{sec:governing-stochastic-model}. Top left: wealth $W_t$ for respectively the cumulative target strategy (orange), its wealth target $\tilde W_t$ (yellow), individual target strategy (green), the optimal static strategy (blue) and cumulative contribution (dark blue). Top right: return of the matching portfolio (blue) and return portfolio (orange). Bottom left: $70\%$ replacement ratio target (yellow) together with the expected replacement ratio of the cumulative target strategy (orange), individual target strategy (green) and optimal static strategy (blue). Bottom right: allocation $1-\alpha_t$ to the matching portfolio for the cumulative target strategy (orange), individual target strategy (green) and optimal static strategy (blue).}
	\label{single_scenario_cum}
\end{figure}
 
The combination strategy is best illustrated by means of the resulting investment decisions, i.e., the optimal control $\alpha_{t,\tau}$ as defined by equation \eqref{eqn:optimal_control}. Figure \ref{fig:decision_combination_strategy} illustrates the optimal allocation to the matching portfolio, $1-\alpha_{t,0}$, for the first contribution $c_0$ as a function of time $t$ and the wealth to target ratio, as defined by \eqref{eqn:state_variable}. In this example, allocations are restricted to multiples of $20\%$. Note that, contrary to the rule-based strategies, in the combination strategy risk can be increased and investments can be transferred from the matching to the return portfolio. All together, this makes the combination strategy more refined than the rule-based strategies, which follow a ``risk on'' or ``risk off'' approach, in terms of their allocation.

%TODO: vraag Tim deze figuur netjes te maken, i.h.b. y-as kan gewoon tot 1 lopen
\begin{figure}[t!]
	\input{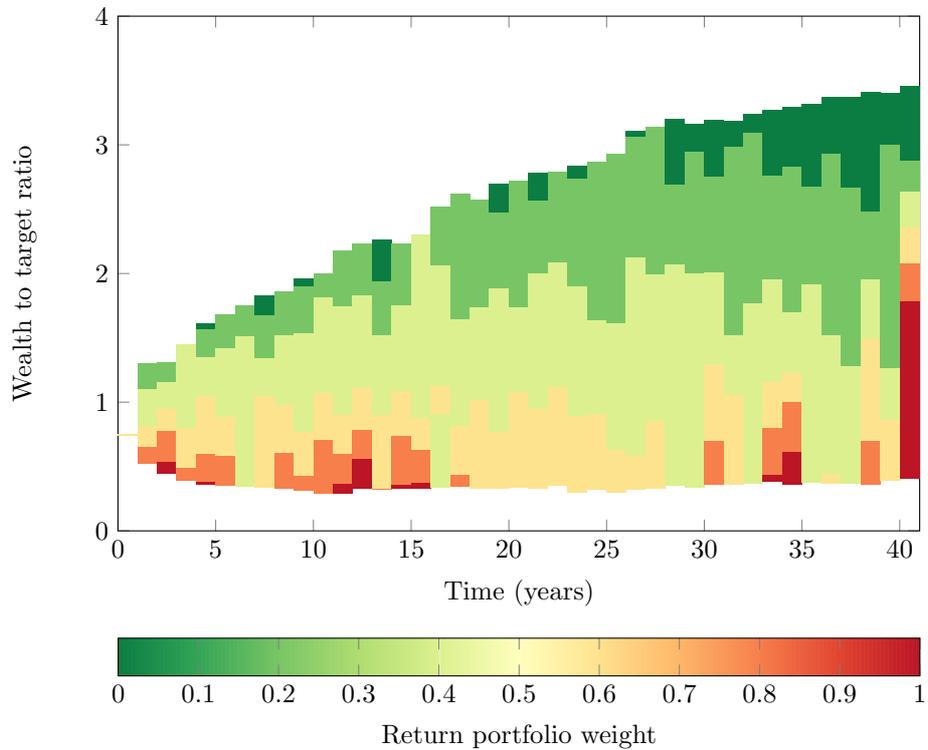}
	\caption{Optimal allocation $1-\alpha_{t,0}$ to the matching portfolio, as a function of wealth to target ratio $\nicefrac{W_{t,0}}{\tilde W_{t,0}}$ (y-axis), as defined by equation \eqref{eqn:state_variable}, and time $t$ (x-axis) for the first contribution of an investor following the combination strategy, discussed in Section \ref{combination_strategy}, and using the  stochastic model of section \ref{sec:governing-stochastic-model}. Allocations are restricted to multiples of $20\%$.}
	\label{fig:decision_combination_strategy}	
\end{figure}

Figure \ref{fig:replacement_ratio} compares the distribution of the terminal replacement ratio for the following best performing strategies in terms of the expected shortfall below the investor's $70\%$ replacement ratio target: two rule-based strategies, a combination strategy and a static strategy. The figure illustrates that, as intended, the dynamic strategies reduce downward risk at the expense of upward potential, i.e., the dynamic strategies are centered more around the target replacement ratio of $70\%$.

%TODO: vraag Tim deze figuur netjes te maken, i.h.b. kleuren hetzelfde als het vorige plaatje en geen legenda
%TODO: vraag Tim om de combinatie strategie in deze figuur aan te passen, deze klopt niet
%TODO: vraag Tim welke parameter waarden hier zijn gebruik te controleren en te kijken of ze overeen komen met de caption, d.w.z. welke is dit op de efficient frontier?
\begin{figure}[pt!]
	\input{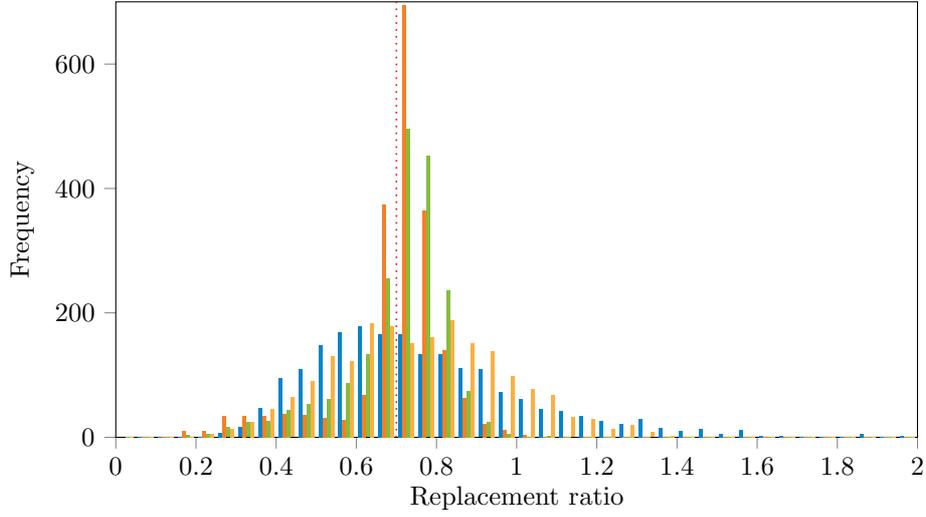}
	\caption{Distribution of the replacement ratio for, respectively, the cumulative target strategy with $r=3.06\%$ (orange), individual target strategy $r=2.99\%$ (green), combination strategy with $r=1\%$ (yellow) and a static strategy with $46.02\%$ constant allocation to the return portfolio.}
	\label{fig:replacement_ratio}	
\end{figure}

A comparison of all strategies is best made by comparing the strategy successes, i.e., whether a strategy achieves the intended $70\%$ replacement ratio target, versus its downside risk, and parametrize the strategies by the parameters that control the strategy's risk appetite, see Figure \ref{fig:10CVaR}. From this figure, we conclude that all the dynamic strategies clearly outperform the traditional static strategies. Together with the intuitive rationale to reduce risk after several good years, we believe this sufficiently demonstrates the added value of these dynamic strategies. We do, however, find these simulations insufficient to rank the dynamic strategies based on their effectiveness. It is well-known that the relative performance of dynamic strategies can be sensitive to the characteristics of the underlying stochastic model. As such, the characteristics are not completely objective, and we believe that use of the strategies in practice is an appropriate way to test the strategies further (which lies beyond the scope of this research).

%Local regressions used in the dynamic programming algorithm are limited to the sample. For the combination strategy, the sample at time $t$ is the set of ratios between portfolio wealth $W_t$ and wealth target $W^*(t)$. Wealth target $W^*(t)$ increases for a higher discount rate, causing the ratio between portfolio wealth and wealth target to decrease. The sample range becomes smaller. Therefore, regressions for high discount rates are more limited compared to regressions for low discount rates. 

%TODO: vraag Tim deze figuur netjes te maken
%TODO: by what did we parametrize the combination strategy?
\begin{figure}
	\input{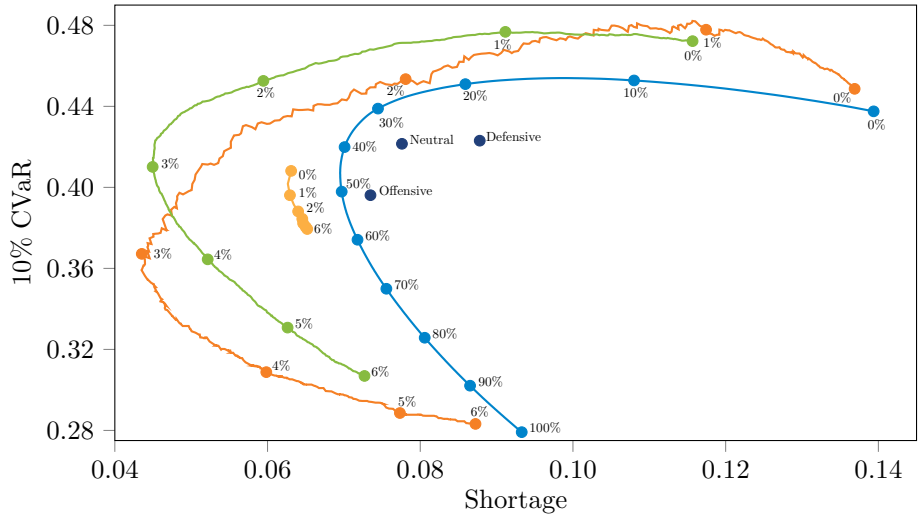}
	\caption{Expected shortfall below the investor's $70\%$ replacement ratio target, see equation \eqref{eqn:utility_shortfall}, versus $10\%$ CVaR of the terminal replacement ratio, as defined by (\ref{kkk}), for: the cumulative target strategy (orange), the individual target strategy (green), the combination strategy (yellow), all parametrized by the real rate of return $r$ (yellow), and, also, several annually rebalanced allocations (blue), and several default life cycles reducing risk with the investor's age.}
	\label{fig:10CVaR}	
\end{figure}

%TODO: hier verder

\subsection{Discussion}
One of the intended advantages of a static life cycle strategy is the reduced risk close to the retirement, meaning that one can provide the investor with an accurate estimate of her retirement income in the years before retirement. Table \ref{table:overview_strategies} provides a comparison of the dynamic strategies and traditional life cycle strategies. In particular, the table lists the standard deviations of the difference between the expected replacement ratio 5 years before retirement and the replacement ratio at retirement. We conclude that when following the rule-based strategies the investor can be provided with a similarly accurate estimate of the replacement ratio before retirement.

%TODO: Tim vragen om deze table te update
\enlargethispage{2.5cm}  % allow two more lines on current page
\begin{table}[H]
\hspace{-1.7cm}
\begin{tabular}{l|lll|lll|lll}
&\multicolumn{3}{c|}{\bfseries\sffamily Static mix} & \multicolumn{3}{c|}{\bfseries\sffamily Static life cycle} & \multicolumn{3}{c}{\bfseries\sffamily Dynamic strategies} \\ 
                    & $0\%$  &$100\%$&$46.02\%$& Def.   & Neut.  & Off.   & Cum.   & Indiv. & Comb.  \\ \hline
Averages ($R$)&&&&&&&&& \\	
\qquad Mean         & $0.56$ & $1.13$ & $0.77$ & $0.66$ & $0.71$ & $0.79$ & $0.70$ & $0.70$ & $0.75$ \\
\qquad $10\%$ CVaR  & $0.44$ & $0.28$ & $0.41$ & $0.42$ & $0.42$ & $0.40$ & $0.36$ & $0.41$ & $0.40$  \\
\qquad $5\%$ CVaR   & $0.42$ & $0.24$ & $0.37$ & $0.39$ & $0.39$ & $0.36$ & $0.28$ & $0.34$ & $0.35$  \\		
Percentiles ($R$) &&&&&&&&&\\		
\qquad Median      & $0.56$ & $0.83$ & $0.72$ & $0.64$ & $0.67$ & $0.72$ & $0.72$ & $0.73$ & $0.75$  \\
\qquad $10\%$ VaR  & $0.46$ & $0.36$ & $0.47$ & $0.47$ & $0.47$ & $0.46$ & $0.51$ & $0.52$ & $0.47$  \\
\qquad $5\%$ VaR   & $0.44$ & $0.29$ & $0.41$ & $0.43$ & $0.44$ & $0.41$ & $0.37$ & $0.44$ & $0.41$  \\		
Goal ($70\%~R$) &&&&&&&&&\\
\qquad Shortage    & $0.139$ & $0.093$ & $0.070$ & $0.088$ & $0.078$ & $0.073$ & $0.044$ & $0.045$ & $0.063$ \\
\qquad Goal reached& $6\%$   & $60\%$  & $53\%$  & $35\%$  & $44\%$  & $53\%$  & $65\%$ & $65\%$ & $57\%$  \\
Estim. error ($R$) &&&&&&&&& \\
\qquad Mean		   & $0.088$ & $0.398$ & $0.139$ & $0.095$ & $0.092$ & $0.115$ & $0.113$ & $0.091$ & $0.144$  \\ 
\qquad Std dev.    & $0.06$  & $0.57$  & $0.12$  & $0.06$  & $0.06$  & $0.10$  & $0.09$  & $0.07$  & $0.12$ \\
\end{tabular}
\caption{Statistics for different investment strategies. Values in this table were calculated using $r = 3.06\%$ for the combination strategy, $r = 2.99\%$ for the individual strategy, and $r = 1\%$ for the combination strategy.}
\label{table:overview_strategies}
\end{table}

Although the rule-based strategies outperform other strategies in our examples, we wish to point that there are also disadvantages in the all-or-nothing approach, e.g, the portfolio remains $100\%$ invested in the more risky return portfolio when targets are not reached. Such truly worst case scenarios appear to have a minor influence, but are, e.g., illustrated in the far left lower tail in Figure \ref{fig:10CVaR}. The individual target strategy presented in Section \ref {sec:indv_target} suffers less from the all-or-nothing disadvantages, as it defines a target per contribution. As a result, inferior past performance does not influence the required performance of current and future contributions.

Compared to the rule-based strategies, the combination strategy does not exploit the fact that the matching portfolio can grow an investment securely to its intended target (indexed by expected inflation) until retirement. As the rule-based strategies explicitly made use of this, the combination strategy could be further improved.

One advantage of the combination strategy in practical use is that the corresponding asset allocation is much more smooth than for the rule-based strategies. The necessity of large turnovers is difficult to explain and investors might be uncomfortable to follow such a drastic strategy to the end.

\section{Conclusion}\label{conclusion}

In this paper, we discussed several dynamic strategies, suitable for pension investors that aim to replace a proportion of their salary with a retirement income. The strategies reduce risk after several good years and steer the investor to her target. By having the allocation depend on return on investment, the approaches exploit a freedom which is typically not used by traditional static approaches. We have shown that the dynamic approaches may outperform some traditional static approaches and prevent unnecessary risk taking.

Two simple and intuitive rule-based strategies were introduced that secure investments in a cash flow matching portfolio once they yielded sufficient return. Although both rule-based strategies can straightforwardly be implemented in practice, we recommend to also investigate alternatives where the investor, e.g., switches between an aggressive traditional life cycle and a matching portfolio to rule out very aggressive portfolios close to retirement.

The rule-based strategies were further refined into a combination strategy based on dynamic programming. In the current setup, the combination strategy may not be superior and we even found that the rule-based strategies outperform the combination strategy in a numerical example. We certainly believe that dynamic strategies based on dynamic programming can be further improved, as also this research clearly demonstrates their added value to pension investors.

A most suitable dynamic strategy is hard to determine objectively as its performance depends on the governing stochastic model. Also, such a dynamic strategy should fit well with practical requirements, such as whether an investor will follow through on the strategy or will feel the need to combine such a strategy with her own judgement, and whether such strategies comply with regulations. This research, however, demonstrates the added value of dynamic strategies to pension investors. In summary, such strategies exploit freedom that is not used by traditional approaches, can steer a pension investor to her target and prevent unnecessary risk taking.

\printbibliography

\appendix

\section{Career path and contribution rate}\label{sec:career-path-and-contribution-rate}

\begin{table}[H]
	\centering
	\begin{tabular}{l|l|l|l|l}
%		{\footnotesize Age }  & {\footnotesize \begin{tabular}{@{}c@{}}Career \\ path\end{tabular} } & {\footnotesize Salary} & {\footnotesize Contribution (\%)} & {\small Contribution} \\
		{\footnotesize Age }  &
		\multicolumn{2}{c|}{{\footnotesize Salary}} & %{\footnotesize Contribution (\%)} & {\small Contribution} \\
		\multicolumn{2}{c}{{\footnotesize Contribution}}\\
		& {\footnotesize Rate }& {\footnotesize Euro} & {\footnotesize Rate }& {\footnotesize Euro} \\
		\hline
		25    & 3\%   & 29403 & 7.8\% & 1270 \\
		26    & 3\%   & 30285 & 7.8\% & 1339 \\
		27    & 3\%   & 31193 & 7.8\% & 1409 \\
		28    & 3\%   & 32129 & 7.8\% & 1482 \\
		29    & 3\%   & 33093 & 7.8\% & 1558 \\
		30    & 3\%   & 34086 & 9.0\% & 1887 \\
		31    & 3\%   & 35108 & 9.0\% & 1979 \\
		32    & 3\%   & 36162 & 9.0\% & 2073 \\
		33    & 3\%   & 37247 & 9.0\% & 2171 \\
		34    & 3\%   & 38364 & 9.0\% & 2272 \\
		35    & 2\%   & 39131 & 10.5\% & 2731 \\
		36    & 2\%   & 39914 & 10.5\% & 2813 \\
		37    & 2\%   & 40712 & 10.5\% & 2897 \\
		38    & 2\%   & 41526 & 10.5\% & 2982 \\
		39    & 2\%   & 42357 & 10.5\% & 3070 \\
		40    & 2\%   & 43204 & 12.2\% & 3670 \\
		41    & 2\%   & 44068 & 12.2\% & 3775 \\
		42    & 2\%   & 44950 & 12.2\% & 3883 \\
		43    & 2\%   & 45849 & 12.2\% & 3993 \\
		44    & 2\%   & 46765 & 12.2\% & 4104 \\
		45    & 1\%   & 47233 & 14.2\% & 4844 \\
	\end{tabular}\hfill
	\begin{tabular}{l|l|l|l|l}
		{\footnotesize Age }  &
		\multicolumn{2}{c|}{{\footnotesize Salary}} & %{\footnotesize Contribution (\%)} & {\small Contribution} \\
		\multicolumn{2}{c}{{\footnotesize Contribution}}\\
		& {\footnotesize Rate }& {\footnotesize Euro} & {\footnotesize Rate }& {\footnotesize Euro} \\
		\hline
		46    & 1\%   & 47705 & 14.2\% & 4911 \\
		47    & 1\%   & 48183 & 14.2\% & 4978 \\
		48    & 1\%   & 48664 & 14.2\% & 5047 \\
		49    & 1\%   & 49151 & 14.2\% & 5116 \\
		50    & 1\%   & 49642 & 16.5\% & 6026 \\
		51    & 1\%   & 50139 & 16.5\% & 6108 \\
		52    & 1\%   & 50640 & 16.5\% & 6190 \\
		53    & 1\%   & 51147 & 16.5\% & 6274 \\
		54    & 1\%   & 51658 & 16.5\% & 6358 \\
		55    & 0\%   & 51658 & 19.4\% & 7476 \\
		56    & 0\%   & 51658 & 19.4\% & 7476 \\
		57    & 0\%   & 51658 & 19.4\% & 7476 \\
		58    & 0\%   & 51658 & 19.4\% & 7476 \\
		59    & 0\%   & 51658 & 19.4\% & 7476 \\
		60    & 0\%   & 51658 & 23.0\% & 8863 \\
		61    & 0\%   & 51658 & 23.0\% & 8863 \\
		62    & 0\%   & 51658 & 23.0\% & 8863 \\
		63    & 0\%   & 51658 & 23.0\% & 8863 \\
		64    & 0\%   & 51658 & 23.0\% & 8863 \\
		65    & 0\%   & 51658 & 26.0\% & 10019 \\
		66    & 0\%   & 51658 & 26.0\% & 10019 \\
	\end{tabular}%
	\caption{Salary in nominal amount, annual salary increase as rate and contribution in both nominal amount and as percentage by age  at $t=0$, corresponding to ultimo 2016. All is based on statistics reported for the Netherlands in 2016. Salary for a 25 year old is reported by Eurostat as the average Dutch salary under 30 years old in 2014 and is indexed with 2 years of wage inflation as reported by the Dutch Central Bureau of Statistics. The salary increases correspond to a career path used in actuarial calculations of the Dutch government (available under document number \href{https://zoek.officielebekendmakingen.nl/blg-230018}{blg-230018}). Contribution percentages are as prescribed under \href{https://wetten.overheid.nl/jci1.3:c:BWBR0039139&z=2017-01-28&g=2017-01-28}{Dutch law} in 2016 and are applied to the difference between the salary and a franchise of 13123 Euro. For $t>0$, nominal amounts are indexed with wage inflation and rates are kept fixed.}
	\label{table_salary}
\end{table}%

\section{Dynamic Programming Implementation} \label{sec:algorithm}

In this section, we will give a brief description of various computational insights from our implementation. 
The algorithm for the combination strategy includes the dynamic programming design, the selection of the state variable and the use local regression.
The local regression technique in this section is also used by the target strategies.

%Consider the dynamic programming procedure from an algorithmic perspective. Solving the sub-problem at time $t_i$ may alter state $Z_{t_{i+1}}$ at $t_{i+1}$. A different solution may then be optimal at time $t_{i+1}$ for a new state $Z_{t_{i+1}}$. 

\subsection{Dynamic Programming Algorithm}

Asset allocations for a dynamic programming strategy follow from Algorithm \ref{algorithm_1}. The algorithm runs backward in time, similar to the algorithms in \citet{Oos16} and \citet{Bin07}. 

\enlargethispage{\baselineskip}

The solution space is discretized such that a suitable algorithm can be used to solve the dynamic programming problem: it is assumed that
\begin{align}\label{eq:discretization}
\S: \T \times \R_{>0} \rightarrow \{a_1, a_2, \ldots, a_k\},
\end{align}
with the values of the control, $a_j \in [0,1]$ for $j \in\{1,\ldots, k\}$. The investor can choose between at most $k$ asset allocations at each time $t$.

Algorithm \ref{algorithm_1} solves the optimal control problem backward in time by calculating the expected utility, $\EX{\check{U}(Z_T)|\F_t}$, with the state variable $Z_t$ (see Section \ref{sec:state-variable}) by the local regression technique (see Section \ref{sec:least-squares-monte-carlo-method}).
The use of local regression is similar to the use of bundling in \citet{Oos16}: neighborhood points are used in the local regressions for each step of the algorithm.

Solving the sub-problem at time $t_i$ changes the future states $Z_{t_{i+1}}, \ldots, Z_{t_{n-1}}$. These states have previously been used in the local regressions to find the optimal solution for the sub-problems at times $t_{i+1}, \ldots, t_{n-1}$. Thus, the optimal solutions for sub-problems at $t_{i+1}, \ldots, t_{n-1}$ may be different after solving the sub-problem at time $t_i$. This is why Algorithm \ref{algorithm_1} follows a ``snake-like pattern'' through time: after the sub-problem at time $t_i$ is solved, future sub-problems are first updated in a forward fashion in time. Sub-problems are subsequently updated backwards in time at the next step until the sub-problem at time $t_{i-1}$ is solved for the first time.

Once all sub-problems have been solved, the solution can be further improved by repeating the procedure. Algorithm \ref{algorithm_1} restarts at the beginning of the snake-like pattern through time. Each iteration of Algorithm \ref{algorithm_1} follows the snake-like pattern from $T$ to $t_0$ once.
\begin{algorithm}
	\DontPrintSemicolon
	\KwData{Scenario Set, asset allocations $a_1, \ldots, a_k$}
	\KwResult{Optimal admissible asset control $\A^*$}
	Initialization: create initial solution by choosing $\alpha_i = a_1$ for all $t_i \in \T = \{t_0, \ldots, t_n = T\}$ for all scenarios.\;	
	\nl\For{$m = 1,\ldots,\mathrm{number~of~iterations}$}{
		\nl\For{$j = 1, \ldots, k$}{\nl Set $\alpha_{n-1} = a_j$ for all scenarios\; 
			\nl Determine the expected utility function $f_{j, n-1}$,\;
			\qquad\qquad $f_{j, n-1} (z) = \EX{\check{U}(Z_T) |Z_{n-1} = z, \alpha_{n-1} = a_j }$,\;
			by using local regression for $Z_{n-1}$ and $\check{U}(Z_T)$ (Least Squares Monte Carlo method)\;
		}
		\nl Determine optimal asset allocation decision at time $t_{n-1}$ for all scenarios:\;
		\qquad\qquad $m = \argmax_{j \in \{1,\ldots, m\}} f_{j, n-1}\lrp{Z_{n-1}}$\;
		\qquad\qquad $\alpha^*_{n-1} = a_m$\;	
		\nl\For{$t_i = t_{n-2}, t_{n-3}, \ldots, t_0$}{
			\nl Update optimal allocation decisions (backward update):\;
			\For{$\tau = n, \ldots, i+1$}{\nl Update expected utility function $f_{j,\tau}$, for $j = 1, \ldots, k$\;
				\nl Update optimal asset allocation decision $\alpha^*_\tau$, for all scenarios\;}	
			\nl Determine optimal asset allocation decision at time $t_i$ for the first time\;	
			\For{$j = 1, \ldots, k$}{\nl Set $\alpha_i = a_j$ for all scenarios\;
				\nl Determine the expected utility function $f_{j,i}$,\;
				\qquad\qquad $f_{j,i} (z) = \EX{\check{U}(Z_T) |Z_{i} = z, \alpha_i = a_j, \A_{i+1} = \A^*_{i+1}}$,\;
				by using local regression for $Z_{i}$ and $\check{U}(Z_T)$ (Least Squares Monte Carlo method)\;}
			\nl  Determine optimal asset allocation decision at time $t$, for all scenarios:\;
			\qquad\qquad $m = \argmax_{j \in \{1,\ldots, k\}} f_{j,i}\lrp{Z_i}$\;
			\qquad\qquad $\alpha^*_i = a_m$\;	
			\nl Update optimal allocation decisions (forward update):\;
			\For{$\tau = i+1, \ldots, n-1$}{\nl Update expected utility function $f_{j,\tau}$, for $j = 1, \ldots, k$\;
				\nl Update optimal asset allocation decision $\alpha^*_\tau$, for all scenarios\;}		
	}}
	\caption{Dynamic programming investment strategy}
	\label{algorithm_1}
\end{algorithm}

\subsection{State Variable}\label{sec:state-variable}
Individual wealth targets are not constant over the time horizon. They are defined using the expected inflation and the market value factor at time $t$ (see also Section \ref{sec:target-replacement-ratio}). The individual wealth targets are also not constant between the different scenarios, as they are dependent on the  contribution. The dynamic programming approach requires, however, a fixed wealth target in order to evaluate the expected utility. This issue is resolved by using the ratio between the portfolio wealth and the wealth target $Z_t$, defined in Equation \eqref{eqn:state_variable}, as the state variable of the dynamic programming algorithm.
%\begin{align*}
%Z_t = \frac{W_t}{W^*(t)}\:.
%\end{align*}
The target state is now constant in time: $Z^*_t = 1$. 

An advantage of choosing this state variable is that, theoretically, the dynamic programming solution only has to be computed once. The investment decisions for the first contribution of the investor can be used for all other contributions. Dynamic programming results for the first contribution span the time horizon $[t_0, T]$. At each rebalancing time, the optimal investment decision, depending on state $Z_t$, has already been computed. Because a ratio is used, these investment decisions can also be used for the contributions in later years.

This is, however, not the case in practice. Algorithm \ref{algorithm_1} does not fully converge due to the limited sample size and due to sets of observations changing multiple times per iteration. Running Algorithm \ref{algorithm_1} separately for each contribution will  give better decision rules. Separate runs of the algorithm will provide a better approximation of the optimal solution on average. 

%The asset allocation of the overall portfolio is no longer bound to the discretization. Instead, the asset allocation of the overall portfolio is now the sum of the asset allocations for each individual contribution.

\subsection{Least squares Monte Carlo method}\label{sec:least-squares-monte-carlo-method}
At each optimization step in Algorithm \ref{algorithm_1}, a least squares Monte Carlo method is used to avoid nested simulations (and the, related, exploding computation times). The least squares Monte Carlo method was introduced by \citet{LonS01} as a simple method for pricing American options by simulation. The conditional expectation of the pay-off under the assumption of not exercising the option is estimated by using cross-sectional information already available in the simulation. Realized pay-offs from continuation (or, in the pension investment setting, from final utility $\check{U}(Z_T)$) are regressed on functions of the state variables. The fitted value provides an estimate of the conditional expectation.

The regression employed is based on a so-called {\em regress-now} strategy, and specifically on a local regression version. Regress-now estimates the expectation $\EX{Y_{t_{i+1}}|X_{t_i}}$, $X_{t_i} \in \F_{t_i}$ by using a set of basis functions, $\Phi$, with index set $\mathcal{J}$:
\begin{equation*}
Y_{t_{i+1}} \approx \sum_{j \in \mathcal{J}} c_j\varphi_j(X_{t_i}),
\end{equation*}
with $c_j$ coefficients found by using least squares regression and $\varphi_j \in \Phi$. Substitution gives
\begin{align*}
\EX{Y_{t_{i+1}}|X_{t_i}} &\approx \EX{\left.\sum_{j \in \mathcal{J}} c_j\varphi_j(X_{t_i})\right|X_{t_i}}\\
&= \sum_{j \in \mathcal{J}} c_j\varphi_J(X_{t_i})\:.
\end{align*}

Local regression, introduced by \citet{Cle79}, estimates a linear or quadratic polynomial fit at $x$ by using a weighted least squares regression. Weights for an observation $(x_i,y_i)$ are dependent on the distance between $x_i$ and $x$ \citep{Cle91}. The smoothness of the fit is dependent on the percentage of observations that are taken into account when evaluating at $x$. 

Let $n$ be the number of observations and let $0<d\leq1$ be a neighborhood parameter, i.e., the share of observations used for the weighted least squares regression at the evaluation point. Furthermore, let $k = d\cdot n$, rounded up to an integer value, $\Delta_i(x)$ be the Euclidean distance of $x$ to $x_i$, and $\Delta_{(i)}(x)$ be the values of these distances, ordered from smallest to largest.

The weight  $\zeta_t$ for an observation $(x_i, y_i)$ is then equal to
\begin{equation*}
 \zeta_i(x) = T\lrp{\frac{\Delta_i(x)}{\Delta_{(k)}(x)}},	
\end{equation*}	
with
\begin{align*}
T(u) &= \left\{ \begin{array}{@{}ll@{}}\lrp{1-u^3}^3, &  \mathrm{for~} 0\leq u < 1, \\ 0, & \mathrm{for~} u \geq 1, \end{array}\right.
\end{align*}
also known as the tri-cube weight function.

 Not only is the regress-now strategy used to approximate the expected utility, this strategy is also used to estimate the expected annual inflation $\mathbb{E}_t\pi_\tau$, with $\Phi = \{1,x\}$. The future cumulative inflation, $Y_{t_{i+1}}$, is regressed on the past cumulative inflation, $X_{t_i}$, to estimate the future annual inflation at time $t_i$:

\begin{align*}
x_j &= \prod_{k = t}^j (1 + \pi_k),\\
y_j &= \prod_{k = j+1}^T (1+ \pi_k),
\end{align*}
with $x_j \in X_{t_i}$ and $y_j \in Y_{t_{i+1}}$ for each $j$ in the scenario set.

The resulting regression function is of the form $f = \xi_1 x + \xi_2$. The expected effective annual inflation for the time period $k$ to $T$ is now equal to:
\begin{equation}
I(T;k) := \sqrt[T - (k+1)]{\xi_1 x_k + \xi_2} -1\:. \label{eqn:inflation}
\end{equation}

\end{document}